\documentclass[11pt]{article}

\usepackage{graphicx}
\usepackage{amssymb,amsmath,amsthm}
\usepackage{fullpage}
\usepackage[left]{lineno}
\usepackage{adjustbox}
\usepackage[utf8]{inputenc}
\usepackage{siunitx}
\usepackage{csquotes}

\newcommand{\comment}[1]{} 

\newtheorem{thm}{Theorem}[section]
\newtheorem{lem}[thm]{Lemma}

\newtheorem{prop}[thm]{Proposition}

\newcounter{quote}

\usepackage{amsfonts}

\usepackage{epsfig}
\usepackage{graphicx}
\usepackage{subcaption}
\usepackage[linesnumbered,ruled,vlined]{algorithm2e}

\usepackage[inline]{enumitem}

\usepackage{authblk}

\newcommand{\RR}{\mathbb{R}}
\newcommand{\Rr}{\mathcal{D}}

\newcommand{\Td}{\mathcal{T}}
\newcommand{\Img}[2]{#1 \left[ #2 \right] }
\newcommand{\set}[1]{\left\lbrace #1\right\rbrace}
\newcommand{\Ar}{\mathcal{A}}

\newcommand{\visp}[4]{\mathcal{V}_{#1}^{#4} \left( #2,#3\right)}
\newcommand{\levp}[3]{{ (\leq #1)level^{#3}} \left( #2\right) }

\newcommand{\mtdu}[2]{\left(#1,#2\right)}
\newcommand{\mttu}[3]{\left(#1,#2, #3 \right)}

\title{Computing $k$-Crossing Visibility through $k$-levels}

\author{Frank Duque\thanks{Email: frduquep@unal.edu.co. Research supported by the Universidad Nacional de Colombia research, grant HERMES-58357. }}
\affil{Escuela de Matem\'aticas, Universidad Nacional de Colombia Sede Medell\'in}

\begin{document}


\maketitle
\begin{abstract}
Let $\Ar$ be a set of straight lines in the plane (or planes in $\RR^3$). The $k$-crossing visibility of a point $p$ on $\Ar$ is the set $Q$ of points in the elements of $\Ar$ such that the segment $pq$, where $q\in Q$, intersects at most $k$ elements of $\Ar$. In this paper, we present algorithms for computing the $k$-crossing visibility. Specifically, we provide $O(n\log n + kn)$ and $O(n\log n + k^2n)$ time algorithms for sets of $n$ lines in the plane and arrangements of $n$ planes in $\RR^3$, which are optimal for $k=\Omega(\log n)$ and $k=\Omega(\sqrt{\log n})$, respectively. We also introduce an algorithm for computing $k$-crossing visibilities on polygons, which achieves the same asymptotic time complexity as the one presented by Bahoo et al. The techniques proposed in this paper can be easily adapted for computing $k$-crossing visibilities on other instances where the $(\leq k)$-level is known.

\end{abstract}


\section{Introduction}\label{sec:intro}
Visibility is a natural phenomenon in everyday life. We use it to determine our location, plan movements without colliding with other objects, or interact with our environment.
Similarly, the concept of visibility is fundamental in several areas of Computer Science, such as Robotics \cite{briggs2000_robot_visibility, Robot_motion_planning2012, lozano1979_collision_free, peshkin1986robots},
Computer Vision \cite{faugeras1993three, Aspect_Graphs1991}, and Computer Graphics \cite{cohen1993radiosity, toth2017handbook, dorward1994survey}. In the case of Computational Geometry, the notion of visibility has also been used extensively in the context of the art gallery problem \cite{o1987art, TOUSSAINT1986165, ghosh2007visibility, shermer1992recent}.

In \cite{aichholzer2018modem}, the authors motivated the study of a new kind of visibility inspired by the development of wireless network connections.
Consider locating a router in a building, so that any device within the building receives a strong enough signal to guarantee a stable Internet connection.
There are two main limitations
when trying to connect a device to a wireless network:
its distance to the wireless router and the number of walls that separate it from the router.
However, in many buildings, the primary limitation is the number of walls that separate the device from the router, rather than the distance between them.
Given a set $\Ar$ of straight lines, rays and segments in $\RR^2$ (or planes in $\RR^3$),  we say that two points $p$ and $q$ are $k$-crossing visible if the interior of the line segment $pq$ intersects at most $k$ elements of $\Ar$.
The set of points in $\RR^2$  ($\RR^3$) that are $k$-crossing visible from $p$ is called the $k$-crossing visibility region of $p$.
Thus, if $p$ and $\Ar$ correspond to the router and the obstacles, respectively, the $k$-crossing visibility region represents the possible positions of the devices, ensuring they are separated by at most $k$ walls. In Figure~\ref{fig:visibility}, we illustrate in yellow the $2$-crossing visibility region of a point on a set of lines.
\begin{figure}[hbt]
\centering
	\includegraphics[width=0.4\textwidth]{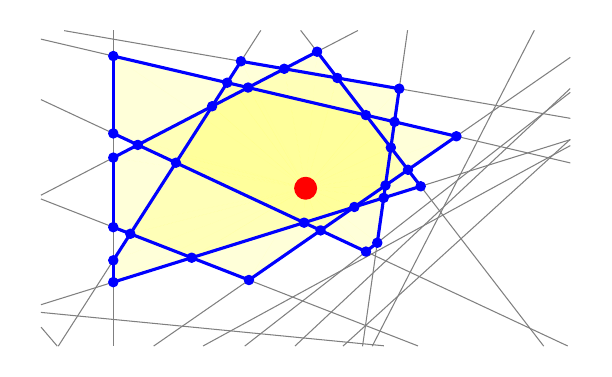}
	\caption{Illustration of the $2$-crossing visibility region (yellow) and the $2$-crossing visibility (blue) of the red point on a set of lines.}
	\label{fig:visibility}
\end{figure}

Let $\Ar$ be a set of straight lines, rays and segments in $\RR^2$ (or planes in $\RR^3$).
The $k$-crossing visibility of a point $p$ in $\RR^2$ ($\RR^3$) on $\Ar$, denoted by $\visp{k}{p}{\Ar}{}$, is the portion of the elements of $\Ar$ that are in the $k$-crossing visibility region of $p$. 
Figure~\ref{fig:visibility} shows $\visp{2}{p}{\Ar}{}$ in blue for a point $p$ and a set of lines $\Ar$.
In this paper, we are interested in algorithms for computing $\visp{k}{p}{\Ar}{}$ and determining upper bounds for its size.

The concept of $k$-crossing visibility was first introduced by Mouawad and Shermer \cite{mouawad1994superman} in what they originally called the Superman problem: Given a simple polygon $P$, a sub-polygon $Q\subset P$, and a point $p$ outside $P$, determine the minimum number of edges of $P$ that must be made opaque such that no point of $Q$ is $0$-crossing visible to $p$. Dean et al. \cite{dean1988recognizing} studied pseudo-star-shaped polygons, in which the line of visibility can cross one edge, corresponding to $k$-crossing visibility where $k=1$. Bajuelos et al \cite{bajuelos2012hybrid} subsequently explored the concept of $k$-crossing visibility for an arbitrarily given $k$, and presented an $O(n^2)$-time algorithm to construct $\visp{k}{p}{P}{}$ for an arbitrary given point $p$.
Recently Bahoo et al.\cite{Computing_k_Visibility_Polygon} introduced an algorithm that computes $\visp{k}{p}{P}{}$ in $O(nk)$ time.

\begin{thm}[Bahoo et al. \cite{Computing_k_Visibility_Polygon}]\label{thm:polygon_visibility}
 Given a simple polygon $P$ with $n$ vertices and a query point $p$ in $P$, the region of $P$ that is $k$-crossing visible from $p$,
 can be computed in $O(kn)$ time.
\end{thm}

In this work, we present another proof of Theorem~\ref{thm:polygon_visibility} and we prove Theorem~\ref{thm:lines_visibility}, Theorem~\ref{thm:planes_visibility}, Proposition~\ref{prop:complexity_lines}, and Proposition~\ref{prop:complexity_planes}.

\begin{thm}\label{thm:lines_visibility}
 Let $\Ar$ be a set of $n$ lines in the plane, and let $p$ be a query point. Then $\visp{k}{p}{\Ar}{}$ can be computed in $O(n\log n + kn)$ time.
\end{thm}

Given a set $\Ar$ of straight lines, rays and segments in the plane (or planes in $\RR^3$), the arrangement defined by $\Ar$ is the subdivision of the plane (space) defined by $\Ar$; in the following we refer to this arrangement as the arrangement of $\Ar$. The combinatorial complexity of an arrangement is its total number of vertices and edges (and faces).

\begin{prop}\label{prop:complexity_lines}
 The maximum combinatorial complexity of the arrangement defined by the $k$-crossing visibility on sets of $n$ straight lines in the plane is $\Theta(kn)$.
\end{prop}

\begin{thm}\label{thm:planes_visibility}
 Let $\Ar$ be an arrangement of $n$ planes in $\RR^3$, and let $p$ be a query point. Then $\visp{k}{p}{\Ar}{}$ can be computed in $O(n\log n + k^2n)$ time.
\end{thm}

\begin{prop}\label{prop:complexity_planes}
 The maximum combinatorial complexity of the arrangement defined by the $k$-crossing visibility on sets of $n$ planes in $\RR^3$ is $\Theta(k^2n)$.
\end{prop}

Note that, by Proposition~\ref{prop:complexity_lines} and Proposition~\ref{prop:complexity_planes}, Theorem~\ref{thm:lines_visibility} and Theorem~\ref{thm:planes_visibility} are optimal for $k=\Omega(\log n)$ and $k=\Omega(\sqrt{\log n})$, respectively. These theorems are also optimal for $k=0$, by Lemma~\ref{lem:lelvels_visibility_paths:R2} and Lemma~\ref{lem:lelvels_visibility_paths:R3}: in this case, the problem is equivalent to computing lower and upper envelopes for sets of lines in the plane or sets of planes in space, which have a time complexity of $\Omega(n\log n)$.

Given a set $\Ar$ of straight lines, rays and segments in the plane (or planes in $\RR^3$), the $(\leq k)$-level-region of $\Ar$ is the set of points in $\RR^2$ (or $\RR^3$) with at most $k$ elements of $\Ar$ lying above it. In the following, we denote by $\levp{k}{\Ar}{}$, the portion of the elements of $\Ar$ that are in the $(\leq k)$-level-region of $\Ar$.

Let $\Td$ be the transformation
\begin{align*}
 \Td\left(\mtdu{x}{y}\right)&=\mtdu{x/y}{1/y}& &\text{ in the $\RR^2$ case, or} \\
 \Td\left(\mttu{x}{y}{z}\right)&=\mttu{x/z}{y/z}{1/z}&  &\text{ in the $\RR^3$ case.}
\end{align*}
In this paper, we obtain a linear time reduction from the problem of obtaining $\visp{k}{p}{\Ar}{}$ to the problem of obtaining $\levp{k}{\Ar}{}$, by applying $\Td$. As a consequence, we obtain another proof of Theorem~\ref{thm:polygon_visibility} and we prove Theorem~\ref{thm:lines_visibility}, Theorem~\ref{thm:planes_visibility}, Proposition~\ref{prop:complexity_lines}, and Proposition~\ref{prop:complexity_planes}. This reduction can be easily adapted for obtaining $k$-crossing visibilities on other instances where the $(\leq k)$-levels are known.

This paper is organized as follows. In Section~\ref{sec:R2}, we first determine several properties of $\Td$, then we use them to prove Theorem~\ref{thm:lines_visibility}, Theorem~\ref{thm:polygon_visibility}, and Proposition~\ref{prop:complexity_lines}. In Section~\ref{sec:R3}, we translate the properties of $\Td$ from $\RR^2$ to $\RR^3$ obtaining Theorem~\ref{thm:planes_visibility} and Proposition~\ref{prop:complexity_planes}.

\section{Results in $\RR^2$} \label{sec:R2}

Let $\Rr$ be the set of points $\mtdu{x}{y}\in \RR^2$ such that $y\neq 0$. Throughout this section, $O$ denotes the point $\mtdu{0}{0}\in \RR^2$ and $\Td$ denotes the transformation $\Td:\Rr\rightarrow \Rr$ defined by $\Td\left(\mtdu{x}{y}\right)=\mtdu{x/y}{1/y}$.

In this section, we prove that $\Td$ establishes a connection between $k$-crossing visibility and $(\leq k)$-levels in $\RR^2$. Then, we use this result to prove Theorem~\ref{thm:polygon_visibility}, Theorem~\ref{thm:lines_visibility}, and Proposition~\ref{prop:complexity_lines}.

\subsection{Properties of $\Td$}

Given $D\subset \Rr$, we denote by $\Img{\Td}{D}$ the image of $D$ under $\Td$. We also denote by $\Img{\Td}{\Ar}$ the set of images of the elements of $\Ar$ under $\Td$. In this section, we first prove several properties of $\Td$. Then, we determine $\Img{\Td}{D}$ for different instances of $D$. Finally, we prove that $\visp{k}{O}{\Ar}{}$ can be obtained from $\levp{k}{\Img{\Td}{\Ar}}{}$.

\begin{prop}
 $\Td$ is self-inverse.
\end{prop}
\begin{proof}
Let $\mtdu{x}{y} \in \Rr$, then $\Td \circ \Td (\mtdu{x}{y})= \Td (\mtdu{x/y}{1/y}) = \mtdu{x}{y}$.
\end{proof}

\begin{prop}\label{prop:lines_to_lines:R2}
 $\Td$ maps straight lines to straight lines. More precisely, if $L$ is the straight line in $\Rr$ with equation $ax+by+c=0$, then $\Img{\Td}{L}$ is the straight line in $\Rr$ with equation $ax+cy+b=0$.
\end{prop}
\begin{proof}
Let $L'$ be the straight line with equation $ax+cy+b=0$. If $\mtdu{x_0}{y_0}\in\Rr$ is in $L$, then $ax_0+by_0+c=0$; thus, as $a\frac{x_0}{y_0}+c\frac{1}{y_0}+b=0$, then $\Td\mtdu{x_0}{y_0}$ is in $L'$. Similarly, if $\mtdu{x_0}{y_0}\in\Rr$ is in $L'$, then $\Td^{-1}\mtdu{x_0}{y_0}=\Td\mtdu{x_0}{y_0}$ is in $L$.
\end{proof}

\begin{prop}
 $\Td$ preserves incidences between points and lines. More precisely, the point $p$ is on the straight line $L$ if and only if $\Td(p)$ is on the straight line $\Img{\Td}{L}$.
\end{prop}
\begin{proof}
Let $p=(x_0,y_0)$ be a point in $\Rr$, and let $L: ax+by+c=0$ be a straight line in $\Rr$. This proof follows from the fact that $(x_0,y_0)$ satisfies $ax+by+c=0$ if and only if $\mtdu{ \frac{x_0}{y_0} }{\frac{1}{y_0}}$ satisfies $\Img{\Td}{L}: ax+cy+b=0$.
\end{proof}

Given a line $L$ in $\Rr$, we denote by $L^+$ the set of points in $L$ whose second coordinate is greater than zero, and we denote by $L^-$ the set of points in $L$ whose second coordinate is less than zero.

\begin{prop}\label{prop:T_preserves_order:r2}
 Let $L$ be a straight line in $\Rr$. Then $\Img{\Td}{L^+}=\Img{\Td}{L}^+$ and $\Img{\Td}{L^-}=\Img{\Td}{L}^-$. Moreover, if $p_1,p_2,\ldots,p_k$ are in $L^+$ (or they are in $L^-$) ordered by their distance to the $x$-axis from the closest to the furthest, then $\Td(p_1), \Td(p_2),\ldots,\Td(p_k)$ are in $\Img{\Td}{L^+}$ (or they are in $\Img{\Td}{L^-}$, respectively), ordered by their distance to the $x$-axis from the furthest to the closest.
\end{prop}
\begin{proof}
As $\Td$ maps straight lines to straight lines and does not change the sign of the second coordinate, then $\Img{\Td}{L^+}=\Img{\Td}{L}^+$ and $\Img{\Td}{L^-}=\Img{\Td}{L}^-$. If the second coordinates of $p_i$ and $p_j$ are $y_i$ and $y_j$, respectively, then the second coordinates of $\Td(p_i)$ and $\Td(p_j)$ are $1/y_i$ and $1/y_j$, respectively. This proof follows from the fact that $|y_i|<|y_j|$ if and only if $|1/y_i|>|1/y_j|$.
\end{proof}

Let $\Rr^+$ denote the set of points in $\Rr$ whose second coordinate is greater than zero, and let $\Rr^-$ denote the set of points in $\Rr$ whose second coordinate is less than zero.
The proofs of Proposition~\ref{prop:segment_to_somethig:R2} and  Proposition~\ref{prop:ray_to_somethig:R2} follows from Proposition~\ref{prop:T_preserves_order:r2}.

\begin{prop}\label{prop:segment_to_somethig:R2}
 Let $D$ be a line segment contained in a straight line $L$, whose endpoints are $p$ and $q$.
 \begin{itemize}
  \item If both $p$ and $q$ are in $\Rr^+$ ($\Rr^-$), then $\Img{\Td}{D}$ is the line segment contained in $\Rr^+$ ($\Rr^-$) whose endpoints are $\Td(p)$ and $\Td(q)$.
  \item If $p$ is on the $x$-axis and $q$ is in $\Rr^+$ ($\Rr^-$), then $\Img{\Td}{D}$ is the ray contained in $\Rr^+$ ($\Rr^-$) defined by the straight line $\Img{\Td}{L}$ and the point $\Td(q)$.
 \end{itemize}
\end{prop}

Given $D\subset \Rr$, we denote by $\overline{D}$ the closure of $D$ in $\RR^2$.

\begin{prop}\label{prop:ray_to_somethig:R2}
 Let $D$ be a non-horizontal ray contained in a straight line $L$, whose endpoint is $p$.
 \begin{itemize}
  \item If both $p$ and $D$ are in $\Rr^+$ ($\Rr^-$), then $\Img{\Td}{D}$ is the line segment contained in $\Rr^+$ ($\Rr^-$) whose endpoints are $\Td(p)$ and the intersection of $\overline{\Img{\Td}{L}}$ with the $x$-axis.
  \item If $p$ is on the $x$-axis and $D$ is in $\Rr^+$ ($\Rr^-$), then $\Img{\Td}{D}$ is the ray defined by the part of the straight line $\Img{\Td}{L}$ in $\Rr^+$ ($\Rr^-$).
 \end{itemize}
\end{prop}

\begin{prop}\label{prop:horizontal_ray_to_somethig:R2}
 Let $D$ be a horizontal ray contained in a straight line $L$ whose endpoint is $p$. If $D$ is contained in $\Rr^+$ ($\Rr^-$), then $\Img{\Td}{D}$ is the horizontal ray in $\Rr^+$ ($\Rr^-$) defined by the straight line $\Img{\Td}{L}$ and the point $\Td(p)$. If $D$ is contained in $\Rr^+$, then $D$ and $\Img{\Td}{D}$ have the same direction; otherwise, $D$ and $\Img{\Td}{D}$ have opposite directions.
\end{prop}
\begin{proof}
If $L$ has equation $by+c=0$, then $\Img{\Td}{L}$ is the horizontal line with equation $cy+b=0$. Thus, if $D$ is contained in $\Rr^+$ ($\Rr^-$), then $\Img{\Td}{D}$ is also contained in $\Rr^+$ ($\Rr^-$); moreover, $\Img{\Td}{D}$ is the horizontal ray defined by the straight line $\Img{\Td}{L}$ and the point $\Td(p)$. As $\Td\mtdu{x_0}{y_0}= \mtdu{\frac{x_0}{y_0}}{\frac{1}{y_0}}$, the direction of $D$ only changes if $D$ is contained in $\Rr^-$.
\end{proof}

From Proposition~\ref{prop:lines_to_lines:R2},  Proposition~\ref{prop:segment_to_somethig:R2},
Proposition~\ref{prop:ray_to_somethig:R2} and Proposition~\ref{prop:horizontal_ray_to_somethig:R2},
we conclude that:
If $D$ is a straight line, ray or segment contained in $\Rr^+$ ($\Rr^-$) then $\Img{\Td}{D}$ is a straight line, ray or segment contained in $\Rr^+$ ($\Rr^-$).

\begin{prop}\label{prop:lines_origin_to_vertical:r2}
 Let $L$ be a straight line in $\Rr$. Then $O\in \overline{L}$ if and only if $\Img{\Td}{L}$ is a vertical line.
\end{prop}
\begin{proof}
This proof follows from the fact that $\overline{L}$ has equation $ax+by=0$ if and only if $\Img{\Td}{L}$ has equation $ax+b=0$.
\end{proof}

Let $\visp{k}{O}{\Ar}{+}$ denote the portions of the elements of $\visp{k}{O}{\Ar}{}$ in $\Rr^+$, i.e.,
\[\visp{k}{O}{\Ar}{+}= \set{D\cap \Rr^+: D\in \visp{k}{O}{\Ar}{} }\]
Similarly, let $\visp{k}{O}{\Ar}{-}$ denote the portions of $\visp{k}{O}{\Ar}{}$ in $\Rr^-$, i.e.,
\[\visp{k}{O}{\Ar}{-}= \set{D\cap \Rr^-: D\in \visp{k}{O}{\Ar}{} }\]
Let $\levp{k}{\Ar}{+}$ denote the portion of the elements of $\levp{k}{\Ar}{}$ in $\Rr^+$, i.e.,
\[ \levp{k}{\Ar}{+} = \set{D\cap \Rr^+: D \in \levp{k}{\Ar}{} } \]
The $(\leq k)$-lower-level-region of $\Ar$ is the set of points in $\RR^2$ ($\RR^3$) with at most $k$ elements of $\Ar$ lying below it.
Let $\levp{k}{\Ar}{-}$ denote the portion of the elements of $\Ar$ in both $\Rr^-$ and the $(\leq k)$-lower-level-region of $\Ar$.

\begin{lem}\label{lem:lelvels_visibility_paths:R2}
 Let $\Ar$ be a set of straight lines, segments, or rays. Then:
 \begin{enumerate}
  \item \label{lem:levels_visibility_regions:R2:1}
  $\visp{k}{O}{\Ar}{+} = \Img{\Td}{ \levp{k}{\Img{\Td}{\Ar}}{+} }$.
  \item \label{lem:levels_visibility_regions:R2:2}
  $\visp{k}{O}{\Ar}{-} = \Img{\Td}{ \levp{k}{\Img{\Td}{\Ar}}{-} }$.
 \end{enumerate}
\end{lem}
\begin{proof}
 We prove \ref{lem:levels_visibility_regions:R2:1}; the proof of \ref{lem:levels_visibility_regions:R2:2} is similar.

Let $p\in\Rr^+$ be such that $p\in D$ for some $D\in \Ar$, and let $L$ be the line that contains $p$ and $O$. Then $p\in L^+$, $\Img{\Td}{D} \in \Img{\Td}{\Ar}$, and $\Td(p) \in \Img{\Td}{D}$. As $\Td$ preserves incidences, by Proposition~\ref{prop:lines_origin_to_vertical:r2} and Proposition~\ref{prop:T_preserves_order:r2}, the line segment between $O$ and $p$ crosses at most $k$ elements of $\Ar$ if and only if there are at most $k$ elements of $\Img{\Td}{\Ar}$ lying above $\Td(p)$.
\end{proof}

\subsection{Proofs of results in $\RR^2$}

\begin{thm}[Everett et al. \cite{lines_levels}]\label{thm:lines_levels}
Let $\Ar$ be a set of $n$ lines in the plane. Then $\levp{k}{\Ar}{}$ can be computed in $O(n \log n + kn)$ time.
\end{thm}

We use Theorem~\ref{thm:lines_levels} to prove Theorem~\ref{thm:lines_visibility}.

\begin{proof}[Proof of Theorem~\ref{thm:lines_visibility}]
Without loss of generality, we may assume that $p$ is at the origin; otherwise, $p$ and the elements of $\Ar$ can be translated. We also may assume that the $x$-axis does not contain an element of $\Ar$ or an intersection between two elements of $\Ar$, otherwise, the elements of $\Ar$ can be rotated.

Note that  $\visp{k}{p}{\Ar}{}$ can be obtained by joining the segments in $\visp{k}{O}{\Ar}{+}$ that have one endpoint on the $x$-axis with the segments in $\visp{k}{O}{\Ar}{-}$ that have the same endpoint on the $x$-axis; the matching can be achieved by marking the lines that contain segments with endpoints on the $x$-axis. Thus,
by Proposition~\ref{prop:lines_to_lines:R2}, $\Img{\Td}{\Ar}$ is a set of $n$ straight lines, and this proof follows from Lemma~\ref{lem:lelvels_visibility_paths:R2} and Theorem~\ref{thm:lines_levels}.
\end{proof}

Let $\Ar$ be a set of straight lines, rays, and segments. The vertical decomposition (also known as trapezoidal decomposition) of $\Ar$ is obtained by erecting vertical segments upwards and downwards from each vertex in $\Ar$ and extending them until they meet another line or all the way to infinity.

\begin{lem}\label{lem:visibility_from_vertical_decomposition:R2}
Let $\Ar$ be a set of $n$ straight lines, rays, and segments. Then $\levp{k}{\Ar}{}$ can be obtained from a vertical decomposition of $\Ar$ in $O(kn)$ time.
\end{lem}
\begin{proof}
Suppose that the vertical decomposition of $\Ar$ is known. Then for each vertex, extend a vertical segment upwards until it reaches $k+1$ elements of $\Ar$ or its way to infinity; such a vertex is in $\levp{k}{\Ar}{}$ if and only if the vertical segment reaches its way to infinity.
\end{proof}

\begin{proof}[Another proof of Theorem~\ref{thm:polygon_visibility}]
As in the proof of Theorem~\ref{thm:lines_visibility}, we assume that $p$ is at the origin and that the $x$-axis does not contain edges or vertices of $P$. By Proposition~\ref{prop:segment_to_somethig:R2}, $\Img{\Td}{P}$ consists of at most $2n$ line segments or rays. Thus, since the $k$-crossing visibility of $O$ on $P$ can be obtained from $\visp{k}{O}{P}{+}$ and $\visp{k}{O}{P}{-}$, by Lemma~\ref{lem:lelvels_visibility_paths:R2} and Lemma~\ref{lem:visibility_from_vertical_decomposition:R2}, it is sufficient to obtain the vertical decomposition of $\Img{\Td}{P}\cap \Rr^+$ and $\Img{\Td}{P}\cap \Rr^-$ in linear time; we do this for $\Img{\Td}{P}\cap \Rr^+$, the other case is similar.

\begin{figure}[htb]
  \centering
  \begin{subfigure}{0.4\textwidth}
      \includegraphics[width=\textwidth]{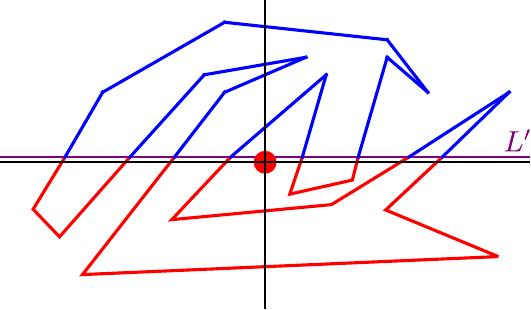}
  \end{subfigure}
  \hspace{2 em}
  \begin{subfigure}{0.4\textwidth}
      \includegraphics[width=\textwidth]{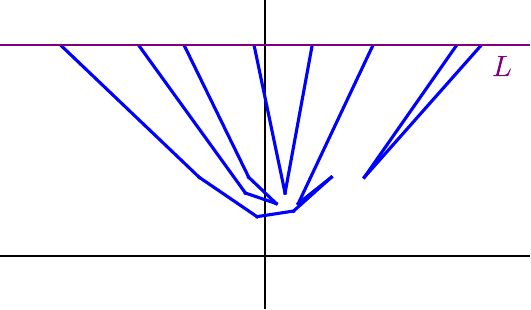}
  \end{subfigure}
  \caption{Illustration of a polygon partition in proof of Theorem~\ref{thm:polygon_visibility}.}
  \label{fig:polygon_partition}
\end{figure}
Let $L: y+c=0$ be a horizontal line, high enough so that all the vertices of $\Img{\Td}{P} \cap \Rr^+$ are below $L$, and
\begin{equation}\label{eq:BoundSlopeBlues}
 \max\set{|x/y|:\ (x,y)\in P } < -c\min\set{|x|:\ (x,y)\in P}.
\end{equation}
Let $L'=\Img{\Td}{L}$, which is a horizontal line with equation $cy+1=0$. Suppose that the points in $P$ above $L'$ are blue and the others are red, see Figure~\ref{fig:polygon_partition}.
Similarly, we define the color of $\Td(q)$ as the color of $q$ for each $q\in P\cap \Rr$. As $L$ is the horizontal line with high $-c$,  $L'$ is the horizontal line with high $1/(-c)$. As there are no vertices of $\Img{\Td}{P} \cap \Rr^+$ above $L$, there are no vertices of $P$ between the $x$-axis and $L'$.
Note that a point $\Td(P)$ is blue if and only if it is between the $x$-axis and $L$. Thus, to obtain a vertical decomposition of $\Img{\Td}{P}\cap \Rr^+$ it is enough to obtain a vertical decomposition of the blue part of $\Img{\Td}{P}$.

\begin{figure}[htb]
  \centering
  \begin{subfigure}{0.4\textwidth}
      \includegraphics[width=\textwidth]{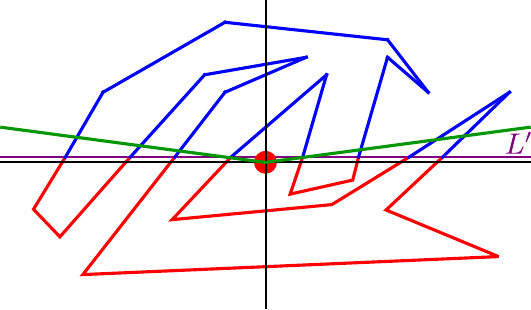}
  \end{subfigure}
  \hspace{2 em}
  \begin{subfigure}{0.4\textwidth}
      \includegraphics[width=\textwidth]{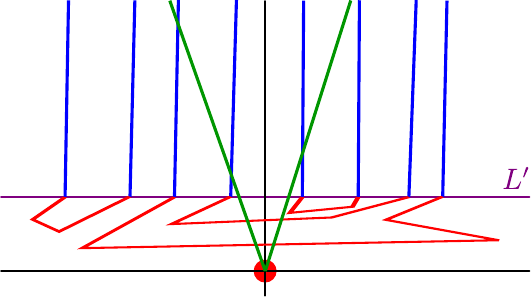}
  \end{subfigure}
  \caption{Illustration for proof of Theorem~\ref{thm:polygon_visibility}.}
  \label{fig:polygon_contraction}
\end{figure}

Let $P'$ be the polygon in $\Rr^+$ obtained from $P$ by vertically scaling its red part while keeping the points on $L'$ fixed (see Figure~\ref{fig:polygon_contraction}).
In \cite{chazelle1991triangulating}, Chazelle proves that the vertical decomposition of a polygon can be computed in linear time (see also Amato et al. \cite{amato2001randomized}). Thus, as $P'$ is contained in $\Rr^+$, by Proposition~\ref{prop:segment_to_somethig:R2}, $\Img{\Td}{P'}$ is a polygon, and the vertical decomposition of $\Img{\Td}{P'}$ can be computed in linear time.

Let \[r=\frac{\min\{|x|:\ (x,y)\in P \text{ and $(x,y)$ is red}\}}{1/(-c)} ,\]
and let $\ell_{left}$ and $\ell_{right}$ be the vertical lines with $x$-coordinates $-r$ and $r$, respectively. We claim that a vertex of $\Img{\Td}{P'}$ is blue if and only if it lies between the lines $\ell_{left}$ and $\ell_{right}$. Specifically, the $x$-coordinate of $\Td(x,y)$ is $x/y$, so any blue point $(x_0,y_0)\in P'$ satisfies
\[ \left|\frac{x_0}{y_0}\right| < \frac{\min\{|x|:\ (x,y)\in P\}}{1/(-c)} \]
and thus lies between $\ell_{left}$ and $\ell_{right}$. In the other case, any red vertex $(x_0,y_0)$ of $P'$ satisfies
\[ |x_0|\geq \min\{|x|:\ (x,y)\in P \text{ and } (x,y) \text{ is red}\} \quad \text{and} \quad 0<y_0<1/(-c), \]
and therefore is not between $\ell_{left}$ and $\ell_{right}$.
Thus, the vertical decomposition of the blue part of $\Img{\Td}{P'}$ can be obtained from the vertical decomposition of $\Img{\Td}{P'}$ below $L$ and between the lines $\ell_{left}$ and $\ell_{right}$.

In order to provide some intuition of the meaning of the lines $\ell_{left}$ and $\ell_{right}$, in Figure~\ref{fig:polygon_contraction} we illustrate $\Img{\Td}{\ell_{left}}$ and $\Img{\Td}{\ell_{right}}$ in green. Note that the blue vertices of $P'$ are above $\Img{\Td}{\ell_{left}}$ and $\Img{\Td}{\ell_{right}}$, and any red vertex of $P'$ is below $\Img{\Td}{\ell_{left}}$ or below $\Img{\Td}{\ell_{right}}$. Finally, note that any point in $\Rr^+$ is to the right of $\ell_{left}$ if and only if it is above $\Img{\Td}{\ell_{left}}$, and it is to the left of $\ell_{right}$ if and only if it is above $\Img{\Td}{\ell_{right}}$.
\end{proof}

\begin{proof}[Proof of Proposition~\ref{prop:complexity_lines}]
In \cite{ALON1986154}, Alon et al. prove that the maximum combinatorial complexity of the arrangement defined by $(\leq k)$-level on sets of $n$ straight lines in the plane is $\Theta(nk)$. Note that the set $\Ar$, which reaches this bound, can be considered with all the points of intersection in $\Rr^+$.
Thus, the combinatorial complexity of the arrangement defined by $\levp{k}{\Ar}{+}$ is also $\Theta(nk)$, and by Lemma~\ref{lem:lelvels_visibility_paths:R2}, the combinatorial complexity of the arrangement defined by $\visp{k}{O}{\Img{\Td}{\Ar}}{+}$ is $\Theta(nk)$.
\end{proof}

\section{Results in $\RR^3$} \label{sec:R3}

Let $\Rr$ be the set of points $\mttu{x}{y}{z}\in \RR^3$ such that $z\neq 0$. Throughout this section, $O$ denotes the point $\mttu{0}{0}{0}\in \RR^3$ and $\Td$ denotes the transformation $\Td:\Rr\rightarrow \Rr$ defined by
\[
\Td\left(\mttu{x}{y}{z}\right)=\mttu{x/z}{y/z}{1/z}.
\]
The proofs of Proposition~\ref{prop:self_inverse:R3}, Proposition~\ref{prop:planes_to_planes:R3}, Proposition~\ref{prop:incidences:R3}, Proposition~\ref{prop:T_preserves_order:r3}, Proposition~\ref{prop:lines_origin_to_vertical:r3}, and Lemma~\ref{lem:lelvels_visibility_paths:R3} can be obtained as in Section~\ref{sec:R2}.

\begin{prop}\label{prop:self_inverse:R3}
 $\Td$ is self-inverse.
\end{prop}
\begin{proof}
Let $\mttu{x}{y}{z} \in \Rr$, then
$\Td \circ \Td (\mttu{x}{y}{z})= \Td (\mttu{x/z}{y/z}{1/z}) = \mttu{x}{y}{z}$.
\end{proof}

\begin{prop}\label{prop:planes_to_planes:R3}
 $\Td$ maps planes to planes. More precisely, if $\pi$ is the plane in $\Rr$ with equation $ax+by+cz+d=0$ then $\Img{\Td}{\pi}$ is the plane in $\Rr$ with equation $ax+by+dz+c=0$.
\end{prop}
\begin{proof}
Let $\pi'$ be the plane with equation $ax+by+dz+c=0$. If $\mttu{x_0}{y_0}{z_0}\in\Rr$ is in $\pi$ then $ax_0+by_0+cz_0+d=0$; thus, as $a\frac{x_0}{z_0}+b\frac{y_0}{z_0}+d\frac{1}{z_0}+c=0$, then $\Td\mttu{x_0}{y_0}{z_0}$ is in $\pi'$. Similarly, if $\mttu{x_0}{y_0}{z_0}\in\Rr$ is in $\pi'$ then $\Td^{-1}\mttu{x_0}{y_0}{z_0}=\Td\mttu{x_0}{y_0}{z_0}$ is in $\pi$.
\end{proof}

\begin{prop}\label{prop:incidences:R3}
 $\Td$ preserves incidences between points and planes. More precisely, the point $p$ is on the plane $\pi$ in $\Rr$ if and only if $\Td(p)$ is on the plane $\Img{\Td}{\pi}$ in $\Rr$.
\end{prop}
\begin{proof}
Let $p=(x_0,y_0,z_0)$ be a point in $\Rr$ and let $\pi: ax+by+cz+d=0$ be a plane in $\Rr$. This proof follows from the fact that $(x_0,y_0,z_0)$ satisfies $ax+by+cz+d=0$ if and only if $\mttu{\frac{x_0}{z_0}}{\frac{y_0}{z_0}}{\frac{1}{z_0}}$ satisfies $\Img{\Td}{\pi}: ax+by+dz+c=0$.
\end{proof}

Given a plane $\pi$ in $\Rr$, we denote by $\pi^+$ the set of points in $\pi$ whose third coordinate is greater than zero, and we denote by $\pi^-$ the set of points in $\pi$ whose third coordinate is less than zero.

\begin{prop}\label{prop:T_preserves_order:r3}
 Let $\pi$ be a plane in $\Rr$. Then $\Img{\Td}{\pi^+}=\Img{\Td}{\pi}^+$ and $\Img{\Td}{\pi^-}=\Img{\Td}{\pi}^-$. Moreover, if $p_1,p_2,\ldots,p_k$ are in $\pi^+$ (or they are in $\pi^-$) ordered by their distance to the plane $z=0$ from the closest to the furthest, then $\Td(p_1), \Td(p_2),\ldots,\Td(p_k)$ are in $\Img{\Td}{\pi^+}$ (or they are in $\Img{\Td}{\pi^-}$, respectively), ordered by their distance to the plane $z=0$ from the furthest to the closest.
\end{prop}
\begin{proof}
As $\Td$ maps planes to planes and it does not change the sign of the third coordinate, then $\Img{\Td}{\pi^+}=\Img{\Td}{\pi}^+$ and $\Img{\Td}{\pi^-}=\Img{\Td}{\pi}^-$. If the third coordinates of $p_i$ and $p_j$ are $z_i$ and $z_j$, respectively, then the third coordinates of $\Td(p_i)$ and $\Td(p_j)$ are $1/z_i$ and $1/z_j$, respectively. This proof follows from the fact that $|z_i|<|z_j|$ if and only if $|1/z_i|>|1/z_j|$.
\end{proof}

Given a set $D\subset \Rr$, we denote by $\overline{D}$ the closure of $D$ in $\RR^3$.

\begin{prop}\label{prop:lines_origin_to_vertical:r3}
   Let $L$ be a straight line in $\Rr$. Then $O\in \overline{L}$ if and only if $\Img{\Td}{L}$ is a vertical line.
\end{prop}
\begin{proof}
As $\Td\mttu{td_1}{td_2}{td_3}= \mttu{\frac{d_1}{d_3}}{\frac{d_2}{d_3}}{\frac{1}{td_3}}$, $\Td^{-1}\mttu{a}{b}{t}=\Td\mttu{a}{b}{t}=\frac{1}{t}\mttu{a}{b}{1}$, $\Td$ is self-inverse and $O\in \overline{L}$ if and only if $\Img{\Td}{L}$ is a vertical line.
\end{proof}

\begin{lem}\label{lem:lelvels_visibility_paths:R3}
 Let $\Ar$ be a set of planes. Then:
 \begin{enumerate}
  \item $\visp{k}{O}{\Ar}{+} = \Img{\Td}{ \levp{k}{\Img{\Td}{\Ar}}{+} }$.
  \item $\visp{k}{O}{\Ar}{-} = \Img{\Td}{ \levp{k}{\Img{\Td}{\Ar}}{-} }$.
 \end{enumerate}
\end{lem}
\begin{proof}
Similar to the proof of Lemma~\ref{lem:lelvels_visibility_paths:R2}.
\end{proof}

The proofs of Theorem~\ref{thm:planes_visibility} and Proposition~\ref{prop:complexity_planes} follow from Theorem~\ref{thm:planes_levels} and Theorem~\ref{thm:levels_hyperplanes}, in a similar way as in the proof of Theorem~\ref{thm:lines_visibility} and the proof of Proposition~\ref{prop:complexity_lines} in Section~\ref{sec:R2}.

\begin{thm}[Chan et al. \cite{planes_levels2}]\label{thm:planes_levels}
 Let $\Ar$ be an arrangement of $n$ planes in $\RR^3$. Then $\levp{k}{\Ar}{}$ can be computed in $O(n \log n + k^2n)$ time.
\end{thm}

\begin{thm}[Clarkson et al. \cite{clarkson1989applications}]\label{thm:levels_hyperplanes}
  Let $k\geq 1$. Then the maximum combinatorial complexity of $(\leq k)$-level on arrangements of $n$ hyperplanes in $\RR^d$ is
  $\Theta \left( n^{ \lfloor d/2 \rfloor } k^{ \lceil d/2 \rceil } \right)$.
\end{thm}


\bibliographystyle{plain}
\bibliography{kvisibility}

\end{document}